\documentclass[letterpaper, 10 pt, conference]{ieeeconf}  

\IEEEoverridecommandlockouts                              
\makeatletter
\let\NAT@parse\undefined
\makeatother

\usepackage{graphicx} 
\usepackage{amsmath} 
\usepackage{amssymb}  
\usepackage{hyperref}
\usepackage{cite}
\usepackage[nameinlink]{cleveref}
\usepackage{xcolor}
\usepackage{algorithm}
\usepackage{algpseudocode}
\usepackage{mathtools}

\newtheorem{theorem}{Theorem}
\newtheorem{remark}{Remark}

\newtheorem{proposition}{Proposition}
\newtheorem{definition}{Definition}

\newcommand{\ie}{\textit{i.e.}}
\newcommand{\eg}{\textit{e.g.}}

\newcommand{\given}{\mid}

\newcommand{\range}[1]{\mathcal{R}(#1)}

\newcommand{\reals}{\mathbb{R}}

\newcommand{\SO}[1]{SO(#1)}

\newcommand{\SE}[2]{SE_{#2}(#1)}

\newcommand{\identity}[1]{I_{#1}}
\newcommand{\zeros}[2]{0_{#1 \times #2}}

\title{\LARGE \bf
Invariant Kalman Filtering with Noise-Free Pseudo-Measurements}

\author{Sven Goffin, Silvère Bonnabel, Olivier Brüls, and Pierre Sacré
    \thanks{Sven Goffin is a FRIA grantee of the Fonds de la Recherche Scientifique - FNRS.}
    \thanks{S.~Goffin and P.~Sacré are with the Department of Electrical Engineering and Computer Science, University of Liège, Belgium (sven.goffin@uliege.be, p.sacre@uliege.be).}
    \thanks{O.~Brüls is with the Department of Aerospace and Mechanical Engineering, University of Liège, Belgium (o.bruls@uliege.be).}
    \thanks{S.~Bonnabel is with MINES ParisTech, PSL Research University, France (silvere.bonnabel@mines-paristech.fr).}
}

\begin{document}

\maketitle
\thispagestyle{empty}
\pagestyle{empty}


\begin{abstract}
    In this paper, we focus on developing an Invariant Extended Kalman Filter (IEKF) for extended pose estimation for a noisy system with state equality constraints. We treat those constraints as noise-free pseudo-measurements. To this aim, we provide a formula for the Kalman gain in the limit of noise-free measurements and rank-deficient covariance matrix. We relate the constraints to group-theoretic properties and study the behavior of the IEKF in the presence of such noise-free measurements.
    We illustrate this perspective on the estimation of the motion of the load of an overhead crane, when a wireless inertial measurement unit is mounted on the hook.
\end{abstract}


\section{Introduction}
Since its introduction in the 1960s, the Extended Kalman Filter (EKF) is probably the most used filter in the industry~\cite{kalman1960new}. Despite its use in many real world applications, it relies on strong assumptions that are rarely met in practice. Moreover, its development is solely based on a probabilistic description of a dynamical and a measurement model, disregarding the geometrical structure of the considered problem. The latter concern gave birth to a proper field  that focuses on geometric filtering methods. This research direction provided improved solutions for many useful tasks, like 
Simultaneous Localization And Mapping (SLAM)~\cite{mahony2017geometric, barrau2018invariant, van2019geometric, mahony2021homogeneous}, 
inertial navigation and localization~\cite{barrau2016invariant,hartley2020contact,barrau2022geometry}, or 
attitude and pose estimation \cite{bonnabel2005invariant, mahony2008nonlinear, cohen2020navigation}. 
The invariance and equivariance properties of frameworks developed in \cite{barrau2015non,barrau2016invariant, mahony2021equivariant, barrau2022geometry} are reminiscent of those of linear systems in many regards. 

In various contexts, some additional (deterministic) information further constrain the state to belong to a subspace of the state space. Incorporating equality constraints in probabilistic filtering constitutes a significant challenge that has led to the realm of state constrained extended Kalman filtering, see~\cite{simon2002kalman}. 
In this paper, we consider nonlinear constraints, dictated by kinematic of mechanical relations, which arise for instance when estimating the configuration of a multi-body mechanical system whose parts are rigidly linked and equipped with one or several Inertial Measurement Units (IMUs). Especially, we consider the equations associated to an IMU that may move freely in space and assume it is mounted at the end of a rigid ``link" whose other end's location is known, or fixed. This may be the case when mounting an IMU on a robotic arm, or on the hook of a crane (a pendulum), as in ~\cite{rauscher2018motion}. In those cases, the known location of the other end of the link provides an (nonlinear) equality constraint that may help estimating the state of the IMU. 

One may then resort to state constrained Kalman filtering, see ~\cite{simon2002kalman}. In the nonlinear case, such methods based on (re)projections are a means of enforcing the constraint. However, they come with no guarantee of consistency whatsoever with the estimation problem at hand, and a brutal (re)projection onto the constraint contains a degree of arbitrariness.  
Herein, we advocate that such nonlinear constraints may rather be incorporated in the extended Kalman filter and its variants through pseudo-measurements without noise, as they constitute information known with certainty. This technique is perfectly justified in the linear case, and allows for optimal fusion of noisy sensors and deterministic side information. This poses two challenges, though. First, the matrix that must be inverted in the Kalman filter may become rank-deficient, see \eg,~\cite{gurumoorthy2017rank}. Second, in the nonlinear case, there is no guarantee the constraint is respected after the pseudo-measurement has been incorporated, contrary to the linear case.  

We derive a formulation of the Kalman gain that accommodates the rank deficiency issues that may stem from the noise-free setting. Furthermore, we cast the problem into the framework of the Invariant Extended Kalman Filter (IEKF) by embedding the state space into a Lie group. Although  \cite{barrau2019extended} provides us with both an inspiration and a mathematical framework, we consider different problems:  \cite{barrau2019extended} considers the limit case of a  noise-free (deterministic) dynamics with noisy measurements, along with an initial condition lying in a constrained set, so that deterministic information is propagated. Herein, we consider arbitrary noisy dynamics (and possibly noisy measurements) starting from an unrestricted initial configuration, and investigate how to incorporate noise-free measurements corresponding to deterministic side information. This comes with different challenges and allows for addressing different problems. Recently, \cite{chauchat2021robust} proposed to linearly constrain the Riccati equation for robust invariant filtering, which is a wholly different problem. 

The paper is structured as follows. 
Section \ref{sec:mot_prob} specifies the considered class of estimation problems,  provides two motivating examples of engineering interest, and  highlights the properties one shall pursue and the challenges that are raised in both the linear and nonlinear cases.
Section \ref{sec:rank_deficient} derives a formulation for the Kalman gain in the limit of noise-free measurements in the linear case and shows its corresponding properties.
Section \ref{sec:nonlinear_nf_update} extends the use of this gain formulation in the framework of invariant Kalman filtering, and proposes an alternative update for the IEKF in the presence of noise-free measurements. Finally, 
Section \ref{sec:simu} illustrates the performance of the proposed approach for the estimation of the extended pose of the hook of a crane.


\section{Motivating problem} \label{sec:mot_prob}

Since the advent of smartphones, the use of cheap IMUs has increased significantly. In this paper, we consider the problem of fusing the IMU measurements with deterministic constraints, treated as noise-free pseudo-measurements. 

\subsection{Mathematical formulation}

Free from any prior knowledge about the motion, the general dynamics of an IMU write (\eg, \cite{barrau2016invariant})
\begin{subequations} \label{eq:sys_dyn}
\begin{align}
    R_{k+1} &
    = R_k\exp((\omega_k + w^{\omega}_k)_\times\, dt),\label{eq:rot_dyn}\\
    v_{k+1} &= v_k + \left(R_k (a_k + w^{a}_k) + g\right)\,dt,\label{eq:vel_dyn}\\
    p_{k+1} &= p_k + v_k\,dt,\label{eq:pos_dyn}
\end{align}
\end{subequations}
where we neglected sensors' biases. The matrix $R_k \in \SO{3}$ is the rotation matrix between the IMU and world frames, and the vectors $v_k,\,p_k \in \reals^3$ are respectively the IMU velocity and position vectors expressed in the world frame at time instant $k$. 
The angular velocity $\omega_k$ and linear acceleration $a_k$ output by the IMU are used as inputs to the dynamical system. 
The vector $g$ denotes the gravity vector expressed in the world frame. 
The map $\exp(\cdot)$ denotes the matrix exponential map and the notation $(b)_\times$ denotes the skew-symmetric matrix associated with cross product with vector $b\in\reals^3$. 

Even if no assumption is made on the dynamics, the structural and mechanical properties of the body to which the IMU is attached may  constrain the IMU pose. 
When dealing with a rigid body, this comes as noise-free equality constraints of the form
\begin{equation}
    R_k r_k + \alpha_k v_k + \beta_k p_k = y_k,
    \label{eq:gen_cst}
\end{equation}
where $r_k,\, y_k \in \reals^3$ and $\alpha_k,\, \beta_k \in \reals$. 

\textbf{Considered problem} We aim at devising a meaningful (Kalman) filter to estimate the state of the noisy system \eqref{eq:sys_dyn} from noise-free measurements \eqref{eq:gen_cst} in real time.

\subsection{Application examples}
Estimating the pose of an IMU under constraints like~\eqref{eq:gen_cst} applies to a wide range of applications. 
We provide two concrete applications for this problem, which is nonlinear, owing  to the state variable $R_k$ being a rotation matrix.

\subsubsection{Crane state estimation}
As a first example, consider a crane on a construction site (see \autoref{fig:crane}). It is technically very feasible nowadays to mount an IMU on the hook that transmits its sensor readings \cite{rauscher2018motion}. Estimating the position and velocity of the load from the IMU measurements may open the door to new automation capabilities since it allows for feedback. However, the underlying dynamical model is unknown, as the hook does not necessarily follow a simple pendulum motion, because of external forces like friction and wind. The IMU dynamics~\eqref{eq:sys_dyn}, though, consist of kinematic relations which are satisfied at all times, whatever the motion. Moreover, when the cable is hanging, the distance between the hook and the cable hang-up point is the cable length~$l_k$, which is very accurately measured in modern cranes for safety reasons, see, \eg, \cite{bonnabel2020industrial}. This provides a constraint of the form of \eqref{eq:gen_cst}
with $r_k=(0 , 0 , -l_k)$, $\alpha_k=0$, $\beta_k = 1$, and $y_k = \zeros{3}{1}$. 

\begin{figure}[htbp]
    \centering
    \includegraphics[width=0.75\columnwidth]{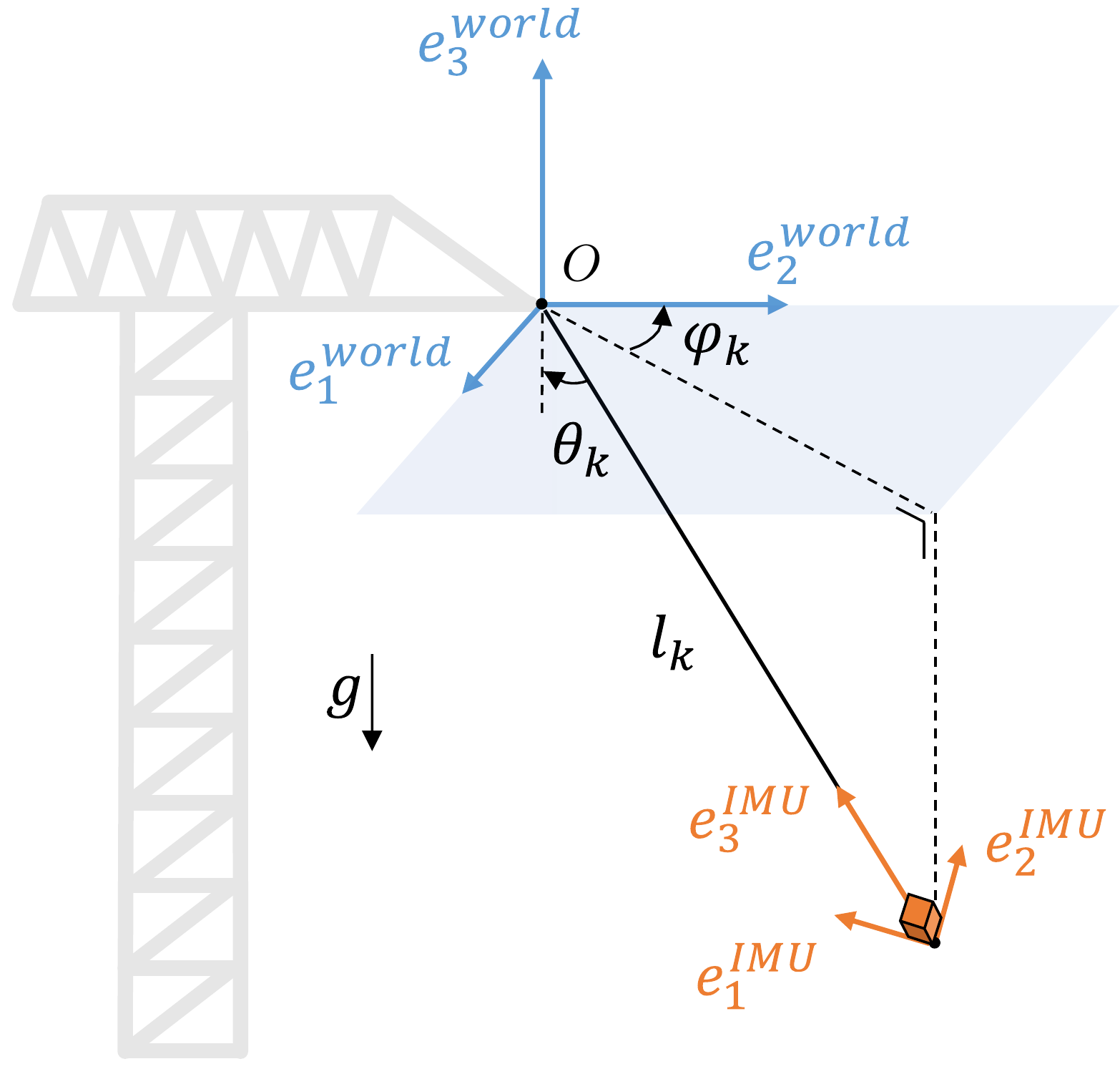}
    \caption{Crane with a wireless IMU mounted on its hook.}
    \label{fig:crane}
\end{figure}

\subsubsection{Robotic arm state estimation}
The field of robotics also uses IMUs extensively. As a second example, robotic arms are usually equipped with IMUs in order to precisely estimate the pose of their end effector. Although the dynamics of the arm cannot be predicted in advance, kinematic constraints enforce that the segments of the arm are attached to each other at the corresponding joints, providing constraints of the form~\eqref{eq:gen_cst}. This is what is exploited in \cite{hartley2020contact} in the form of noisy measurements.

\subsection{Kalman filtering with noise-free pseudo-measurements}\label{subsec:desirable}
Our problem of estimating noisy systems while treating noise-free equality constraints like \eqref{eq:gen_cst} fits into the framework of probabilistic filtering. 
In this framework, the belief at time index $k$, taking into account measurements up to time $j$, consists of a multivariate Normal distribution parameterized by a mean $\hat{x}_{k \given j} \in \reals^n$ (the state estimate) and a covariance matrix $P_{k \given j} \in \reals^{n \times n}$ (the associated uncertainty). 
We advocate indeed that the constraints can be incorporated in a probabilistic filter by using $y_k$ as a noise-free pseudo-measurement $y_k=h(x_k)$, letting 
 \begin{equation}
     h(x_k) = h(R_k,v_k,p_k)=R_k r_k + \alpha_k v_k + \beta_k p_k.
     \label{eq:meas_fct}
 \end{equation}
From the recent theory of two-frames \cite{barrau2022geometry}, it can readily be checked \eqref{eq:meas_fct} defines an output that fits into the framework of invariant filtering, and so do the IMU equations \eqref{eq:sys_dyn}.

Noise-free pseudo-measurements convey information that is known for certain. Hence there are two desirable properties one would expect from a (Normal) probabilistic filter.
\begin{itemize}
    \item 
    \textit{Property 1}: 
    After a noise-free pseudo-measurement $y_k$, the state is known to belong to the subset 
    $$\{x \in \reals^n \mid h(x) = y_k\}$$ 
    owing to noise-free measurement. 
    This should translate into $h(\hat{x}_{k \given k}) = y_k$.
    \item  
    \textit{Property 2}: 
    After a noise-free pseudo-measurement $y_k$, the error covariance matrix $P_{k \given k}$ should encode some deterministic information: null variance is to be expected in the perfectly observed directions. Mathematically, it should translate into $H_kP_{k \given k}H_k^T=\zeros{m}{m}$ with $H_k \in \reals^{m\times n}$ the Jacobian matrix of $h$ at the estimate~$\hat{x}_{k \given k}$.
    
\end{itemize}

\subsection{The linear case}
In the linear case, the Kalman filter---whose equations are recalled in Appendix~\ref{app:KF}---is optimal. Thus, we can expect both properties above to be satisfied. We have
\begin{equation}
H_k\hat{x}_{k \given k}=H_kK_ky_k+(H_k-H_kK_kH_k)\hat{x}_{k \given k-1} \label{eq:linearoutput}
\end{equation}
and 
\begin{equation}
    H_kP_{k \given k}H_k^T=(H_k-H_kK_kH_k)P_{k \given k-1}H_k^T.
\end{equation}
Both desirable properties above are satisfied indeed, as soon as $H_kK_k=\identity{m}$, since it ensures that $H_k-H_kK_kH_k=\zeros{m}{n}$. Recalling that $K_k=P_{k \given k-1}H_k^T(H_kP_{k \given k-1}H_k^T)^{-1}$ whenever measurements are noise-free, the desired relation holds if $H_kP_{k \given k-1}H_k^T$ is invertible. 
If it is not, the Kalman gain is not even defined. 
This may occur when a noise-free measurement is made for the second time: the rank of $H_kP_{k \given k-1}H_k^T$ has dropped and the matrix cannot be inverted anymore. 
This may also occur when performing a subsequent noise-free measurement along directions that overlap with the span of $H_k$ so that the matrix to be inverted is not full-rank. 

\subsection{The nonlinear case}

In the nonlinear case, that is, when using an Extended Kalman Filter (EKF), problems accumulate. Indeed, even if $H_kP_{k \given k-1}H_k^T$ is invertible leading to $H_kK_k=\identity{m}$, there is no reason one should have $h(\hat{x}_{k \given k})=y_k$, as we see that \eqref{eq:linearoutput} relies on $h$  being linear, which is not the case in our problem, see \eqref{eq:meas_fct}, as $R_k$ is a rotation matrix that cannot be treated linearly. Moreover, the Jacobian matrix $H_k$ depends on the linearization point in EKF design, so that the desired relation $H_kP_{k \given k}H_k^T=\zeros{m}{m}$ is not properly defined: the Jacobian matrices at different linearization points span different directions, making the condition unclear, in the sense that the dispersion encoded by $P_{k \given k}$ has no clear relation to the subset defined by $h(x)=y_k$.

In the sequel, we address the issues that arise both in the linear and nonlinear case. 


\section{Handling rank deficiency in the linear case} \label{sec:rank_deficient}

Let us consider the discrete-time linear system
\begin{subequations}
\begin{align}
    x_{k+1} &= F_{k} x_{k} + B_k u_{k}+ w_k,\\
    y_{k} &= H_kx_{k}.\label{eq:noisefree:output:linear}
\end{align}
\end{subequations}
where $u_k\in \reals^b$ is the system input, and where $F_k \in \reals^{n\times n}$, $B_k\in \reals^{n\times b}$, and $H_k \in \reals^{m\times n}$. 
 When the measurement~$y_k$ is noise-free, as in \eqref{eq:noisefree:output:linear}, the Kalman gain \eqref{eq:Kalman:gain} is not defined if $H_kP_{k \given k-1}H_k^T$ is singular, since the measurement noise covariance $N_k$ is null. This is logical as then the same measurements does not bring additional information. However, in the noisy setting we may need to update the state regularly to combat the dispersion due to noise. Then $H_kP_{k \given k-1}H_k^T$ may be nearly singular, and we may run into numerical issues. We now address this problem.

 \subsection{Solving for the Kalman gain in the noise-free limit}
 
 The following result proves the limit of the Kalman gain when $N_k\to 0$ is well defined and provides its expression.
\begin{theorem}
    Assume $P_{k \given k-1}$  to be of rank $l\leq n$ and write it as   $P_{k \given k-1} = L_kL_k^T$ where $L_k \in \reals^{n\times l}$ has linearly independent columns. The following Kalman gain  
    \begin{equation}
        K_k := L_k (H_k L_k)^\dagger,\label{eq:limit;gain}
    \end{equation}
    is the limit of $P_{k \given k-1}H_k^T(H_kP_{k \given k-1}H_k^T+N_k)^{-1}$ as the measurement noise covariance matrix $N_k$ shrinks to $\zeros{m}{m}$, where $A^\dagger$ denotes the Moore-Penrose pseudo-inverse of $A$. 
    \label{theo:nf_gain}
\end{theorem}

Note that the given expression is valid whatever the invertibility of the innovation covariance $H_kP_{k \given k-1}H_k^T$ as $(H_kL_k)^\dagger$ always exists. 

\begin{proof}
Given any matrix $A$, its Moore-Penrose pseudo-inverse can be expressed as the limit
\begin{equation}
    A^\dagger = \lim_{\delta \rightarrow 0} A^T (AA^T + \delta \, I)^{-1}.
    \label{lem:lim_pinv}
\end{equation}
The limit is finite and exists even if $AA^T$ is singular \cite{golub2013matrix}. We pose $A_k=H_kL_k$ and assume for now that $N_k = \delta \, \identity{m}$. 
\begin{align}
    \lim_{\delta \rightarrow 0} K_k &= \lim_{\delta \rightarrow 0} P_{k \given k-1}H_k^T(H_kP_{k \given k-1}H_k^T + \delta \, \identity{m})^{-1},\\
    &= \lim_{\delta \rightarrow 0} L_kA_k^T(A_kA_k^T + \delta \, \identity{m})^{-1},\\
    &= L_k \,\lim_{\delta \rightarrow 0} A_k^T\left(A_kA_k^T + \delta \, \identity{m}\right)^{-1},\\
    &\stackrel{\eqref{lem:lim_pinv}}{=} L_kA_k^\dagger,\\
    &= L_k(H_kL_k)^\dagger.
\end{align} When $N_k$ is not of the form $\delta \, \identity{m}$, we may upper and lower bound its eigenvalues by matrices of the desired form, and use the squeeze theorem (in the sense of positive semidefinite matrix inequalities), using the fact matrices $A_kA_k^T$ and $N_k$ are symmetric. 
\end{proof}

\subsection{Properties of the Riccati update (Property 2)}

When a noise-free measurement \eqref{eq:noisefree:output:linear} has been made, it is logical that the corresponding variance drops to zero, that is, $H_kP_{k \given k}H_k^T=\zeros{m}{m}$, along the lines of Property 2. We now show this is the case with the obtained limit gain. 

\begin{theorem}
    Assume $P_{k \given k-1}$  to be of rank $l\leq n$ and write it as $P_{k \given k-1} = L_kL_k^T$ where $L_k \in \reals^{n\times l}$ has linearly independent columns. Consider the limit gain \eqref{eq:limit;gain} and the corresponding Riccati update \eqref{eq:ric:up}. Then we have $H_kP_{k \given k}H_k^T=0$ or equivalently $H_kP_{k \given k}=0$.\label{thm22}
\end{theorem}
\begin{proof}
    We have 
    \begin{align}
        & H_k(I-K_kH_k)P_{k \given k-1} \\
        & \quad = H_k(I-L_k(H_kL_k)^\dagger H_k)P_{k \given k-1} \\
        & \quad = (H_kL_k-H_kL_k(H_kL_k)^\dagger H_kL_k)L_k^T\\
        & \quad =\zeros{m}{n},  
    \end{align}
     where we used the property of the pseudo-inverse that that $AA^\dagger A=A.$  This proves the result.
\end{proof}


\section{Invariant Kalman filtering with noise-free pseudo-measurements} \label{sec:nonlinear_nf_update}
Consider the discrete-time nonlinear system \eqref{eq:sys_dyn} along with nonlinear output map \eqref{eq:meas_fct}. The EKF first linearizes the system before applying the machinery of linear Kalman filtering. Gain formulation \eqref{eq:limit;gain} may thus be applied in the nonlinear case too. Nevertheless, the dependence of Jacobian matrix $H_k = \frac{\partial h(x)}{\partial x}$ on the linearization point makes the interpretation obscure: the update ensures Property 2 at the linearization point, which is not the true state. The invariant filtering framework brings clarification in this regard.

\subsection{The Invariant EKF for the considered problem}
Let us modify the representation to cast the problem within the invariant framework of  \cite{barrau2018invariant}. We embed the state into the matrix Lie group $\SE{3}{2}$, namely the group of extended poses \cite{barrau2016invariant, brossard2021associating}. The state turns into
\begin{equation}
    \chi_k = \begin{bmatrix}
    R_k & v_k & p_k\\
    \zeros{1}{3} & 1 & 0\\
    \zeros{1}{3} & 0 & 1
    \end{bmatrix}.
\end{equation}
Its dynamics satisfies Equation~(11) of \cite{barrau2018invariant} and is said to be group affine, insuring it fits the framework of the Invariant Extended Kalman Filter (IEKF). In this embedding, the measurement function may write as
\begin{align}
    h(\chi_k) &= \chi_k d_k \label{eq:iekf_meas}
    = 
    \begin{bmatrix}
        R_k r_k + \alpha_k v_k + \beta_k p_k\\
        \alpha_k\\
        \beta_k
    \end{bmatrix},
\end{align}
where $d_k = (r_k , \alpha_k , \beta_k)$, which fits into the theory of \cite{barrau2016invariant} and corresponds to left-invariant outputs (observations in the fixed frame). Along the lines of \cite{barrau2016invariant} we define the innovation as $z_k=\hat{\chi}_{k \given k-1}^{-1}y_k - d_k$. 
The invariant framework introduces the nonlinear and linearized errors $\eta_{k  \given  j}\in \SE{3}{2}$ and $\xi_{k  \given  j}\in \reals^9$ defined by
\begin{align}
    \eta_{k  \given  j} & = \hat{\chi}_{k  \given  j}^{-1}\chi_k = \exp(\xi_{k  \given  j}),
\end{align}
where $\exp(\cdot)$ now denotes the Lie exponential map
of group $\SE{3}{2}$. To derive the Jacobian matrices, the fastest way is to readily retrieve them from the formulas of the recent two-frames theory \cite{barrau2022geometry}. In this context the left-invariant error is $\left(\hat{R}_{k  \given  j}^{-1}R_k, \hat{R}_{k  \given  j}^{-1}(v_k-\hat{v}_{k  \given  j}),\hat{R}_{k  \given  j}^{-1}(p_k-\hat{p})\right)$, 
the innovation is $z_k=(\hat{R}_{k  \given  k-1}^{-1}R_k - I)r_k+\alpha_k \hat{R}_{k  \given  k-1}^{-1}(v_k-\hat{v}_{k  \given  k-1})+\beta_k \hat{R}_{k  \given  k-1}^{-1}(p_k-\hat{p}_{k  \given  k-1})$ and the theory yields, see Proposition 13 of \cite{barrau2022geometry},
\begin{align}
    H_k = 
    \begin{bmatrix}
        -( r_k)_\times &\alpha_k \identity{3} &   \beta_k \identity{3}    
    \end{bmatrix},\label{Jac:lie}
\end{align}
and denoting $\Omega_k:=\exp(dt(\omega_k)_\times)$, we have as in \cite{barrau2022geometry}
\begin{align}
    F_k = 
    \begin{bmatrix}
        \Omega_k^{-1}&0&0\\
        -dt\,\Omega_k^{-1}(a_k)_\times& \Omega_k^{-1}&0 \\
        0& dt \,\Omega_k^{-1}&  \Omega_k^{-1}
    \end{bmatrix}.
\end{align}
We see we recover the important property of IEKF that the Jacobian matrices are state-independent. 
The IEKF updates its estimate according to
\begin{equation}
    \hat{\chi}_{k \given k} = \hat{\chi}_{k \given k-1} \exp(K_kz_k).
    \label{eq:iekf_upd}
\end{equation}

The following sections presents how the invariant framework naturally ensures Property~2 and what can be done to satisfy at best Property~1. Since the problems tackled here arise from the noise-free nature of the considered pseudo-measurements, our work constitutes already a contribution in its own right compared to \cite{barrau2018invariant}, as they only consider noisy measurements.

\subsection{What about desirable Property 2?}

In the formalism of invariant filtering, the constraint \eqref{eq:meas_fct}
writes $\chi_k d_k=y_k$. This noise-free pseudo-measurement informs us that the state belongs to the subset
\begin{equation}
        \mathcal{H} = \{\chi \in \SE{3}{2} \mid \chi d_k = y_k \, \}.
\end{equation}

We first recall how the IEKF encodes uncertainty, see \cite{barrau2018invariant}.
\begin{definition}
The IEKF belief at time instant $k$, accounting for  measurements up to time $j$, for the true state $\chi_k$ is a concentrated Gaussian distribution on the Lie group \cite{bourmaud2015continuous}
\begin{align}
    \chi_{k}=\hat\chi_{k \given j}\exp(\xi_{k \given j}),\quad \xi_{k \given j}\sim\mathcal N(\zeros{n}{1},P_{k \given j}),\label{eq:uncertainty}
\end{align}
with $\hat\chi$ the ``noise-free" mean estimate, and $P$ the covariance matrix encoding the statistical dispersion around the mean. \end{definition}
As a result, for the belief to be consistent with the set $\mathcal{H}$, we would ideally like the entire updated distribution to lie within $\mathcal{H}$. It turns out to be the case indeed, as long as the updated mean is in the right subspace of course. This proves the structure of the invariant filtering framework is consistent with the physical problem indeed. 
 
\begin{theorem}
    Consider the IEKF's covariance matrix $P_{k \given k}$ after updating its state via observation \eqref{eq:meas_fct}. Then if the updated estimate $\hat{\chi}_{k \given k}$ satisfies $\hat{\chi}_{k \given k}\in\mathcal H$, all the probabilistic dispersion encoded by the updated covariance matrix according to our belief model \eqref{eq:uncertainty} lies within $\mathcal H$, that is, for all random $
        \xi\in\mathcal N(0,P_{k \given k})$, we have $  \hat{\chi}_{k \given k}\exp(\xi)\in \mathcal H $.  
\end{theorem}
\begin{proof}
The result is a direct application of the following property from \cite{chauchat2017kalman,barrau2019extended}, which can be proved through a series expansion of the exponential, see  the proof of (iv) of Theorem 1 in \cite{chauchat2017kalman}.
\begin{proposition}\label{proposition1}
    Consider an element $\zeta$ of the Lie algebra of $\SE{3}{2}$ identified with $\mathbb R^9$. Let $H_k$ denote the Jacobian from invariant filtering associated to map $h(\chi):=\chi d_k$, i.e. given by \eqref{Jac:lie} in our problem. Then we have
 $
\left(\chi\in\mathcal{H}~\text{ and }~H_k\zeta=0\right)\Rightarrow \chi\exp\zeta\in \mathcal{H}.$
\end{proposition}

Now, the Riccati update with the limit gain \eqref{eq:limit;gain} ensures that $H_kP_{k \given k}=\zeros{m}{n}$, see Theorem \ref{thm22}. Thus any element $\xi$ in the span of $P_{k \given k}$ satisfies $H_k\xi=0$
 \end{proof}

We thus see that Property 2 is satisfied in a way that is meaningful and leads to a consistent belief: the  covariance matrix $P_{k \given k}$ being aligned with the actual dispersion, regardless of the linearization point. This is a clear improvement over the conventional EKF, that possesses none of the desired properties. In practice, this explains why the IEKF outperforms the EKF for the considered problem, as can be seen below. But before, let us turn to Property 1. 

\subsection{What about Property 1?}

In the previous theorem we had to assume Property~1 to be satisfied, that is, $\hat\chi_{k \given k}\in\mathcal{H}$ for Property~2 to make complete sense. Unfortunately, due to the nonlinear nature of the problem, the IEKF does not satisfy Property~1. All we know is that the residual innovation (in other terms the prediction error) is null up to second order terms after the update, that is,
\begin{align}
    z_{k \given k} &= \hat{\chi}_{k \given k}^{-1}y_k - d_k = \mathcal{O}\left( \| \xi_{k \given k} \| ^2\right),
\end{align}and as a result  $\hat{\chi}_{k \given k}$ may not belong to the expected set $\mathcal{H}$.

To address this issue, we now propose an alternative IEKF update procedure to reduce the norm of the residual innovation. The rationale is the following. Pseudo-measurements come for free, in the sense that they are not actual sensor measurements,  but side information that one has for certain. In principle, once a noise-free measurement is made, the drop in uncertainty is fully incorporated in the belief, so that making the same measurement immediately does not change the belief. In practice, we have seen this is not the case, owing to undesirable effects of nonlinearity. However, there is no reason one could not reuse this same information several times in a row. 

Our idea is thus to update the estimate with the same noise-free pseudo-measurement $y_k$ and the same Kalman gain $K_k$ as long as it makes the prediction error decrease, in other words cycle on the noise-free measurement until $ \| z_{k \given k} \| $ stabilizes. This procedure makes sense if we keep the gain $K_k$ constant when performing the same noise-free measurement several times, since as soon as $K_k$ is updated it does not correct the state along the measurement direction, owing to Property 2 (as $K_kz_k$ lies in the span of $P_{k \given k}$).  The corresponding procedure is detailed in Algorithm~\ref{algo:iter_upd}.

\begin{algorithm}
    \caption{Update procedure for the IEKF with noise-free pseudo-measurements.}\label{algo:iter_upd}
    \begin{algorithmic}[1]
    \State Compute $K_k$ using \eqref{eq:limit;gain}
    \State $z_k^0 \gets \hat{\chi}_{k \given k-1}^{-1}y_k - d_k$
    \State $\chi \gets \hat{\chi}_{k \given k-1} \exp(K_k z_k^0)$
    \State $z_k^1 \gets \chi^{-1}y_k - d_k$
    \State $i \gets 1$
    \While{$ \| z_k^i - z_k^{i-1} \|  > \mathrm{tol}$}
        \State $\chi \gets \chi\exp(K_k z_k^i)$
        \State $z_k^{i+1} \gets \chi^{-1}y_k - d_k$
        \State $i \gets i+1$
    \EndWhile
    \State $\hat{\chi}_{k \given k} \gets \chi$
    \State $P_{k \given k} \gets (I - K_kH_k)P_{k \given k-1}(I - K_kH_k)^T$
    \end{algorithmic}
\end{algorithm}

\begin{remark}Algorithm~\ref{algo:iter_upd} reduces the residual innovation but is not always able to completely eliminate it. Indeed, the innovation can be decomposed as
\begin{equation}
    z_k = z_{k,\parallel} + z_{k,\perp},
\end{equation}
where $z_{k,\parallel}\in \range{S_k}$ and $z_{k,\perp} \in \range{S_k}^\perp$, with $z_{k,\perp} = \mathcal{O}\left( \| \xi_{k \given k-1} \| ^2\right)$ being entirely due to linearization errors and $S_k=H_kP_{k \given k-1}H_k^T$. The component $z_{k,\perp}$ lies in the kernel of $K_k$, \ie, $K_kz_{k,\perp} = \zeros{n}{1}$. This means this part of the innovation accounts for a component $\xi_{k,\perp}$ of $\xi_{k \given k-1}$ that lives within $\range{P_{k \given k-1}}^\perp$ and cannot be corrected using Kalman filtering techniques. As a consequence, the updated state $\hat{\chi}_{k \given k}$ becomes slightly inconsistent with $P_{k \given k}$. In most applications, this inconsistency issue is immediately fixed at the propagation stage as process noise reintroduces uncertainty in the directions for which the last noise-free update negated the variance. In the absence of process noise, this is problematic as it makes the filter become quickly overconfident. This problem is left unsolved here and will be the subject of future work.
\end{remark}

\section{Estimation of the pendulum angle for a crane} 
\label{sec:simu}

We consider the problem of estimating the extended pose (orientation, position, velocity) of an IMU fixed on the hook of a crane as presented in \autoref{sec:mot_prob} and illustrated in \autoref{fig:crane}, and compare the three following  filters:
\begin{itemize}
    \item \textit{EKF}: a conventional extended Kalman filter.
    \item \textit{IEKF}: a conventional invariant extended Kalman filter.
    \item \textit{Noise-free IEKF}: the invariant extended Kalman filter that implements gain \eqref{eq:limit;gain} and uses  Algorithm~\ref{algo:iter_upd} to further refine its estimate at each pseudo-measurement.
\end{itemize}

For the conventional EKF and IEKF, the gain was computed by setting the noise covariance $N_k$ to the small value of $10^{-4}\identity{2}$ and letting $K_k = P_{k \given k-1}H_k^T(H_kP_{k \given k-1}H_k^T + N_k)^{-1}$.
Estimating the IMU pose in the 3D space with cable length as the unique constraint raises observability issues that might obscure the big picture of the present paper. For the sake of simplicity, the IMU is assumed to stay in the plane $\phi_k=0$ so that we consider the associated 2D problem, letting $\chi_k \in \SE{2}{2}$ and  $\xi_{k \given j}\in \mathbb{R}^{5}$. 

The ground-truth simulation is the trajectory followed by the hook when it starts with an initial angle $\theta_0 = 20^{\circ}$ and no angular velocity. No other external force than gravity is simulated. The crane cable is assumed to stay straight during the entire simulation and its length varies according to the profile displayed in \autoref{fig:length_cable}. The simulation is stopped after $2$~s.
\begin{figure}[thpb]
    \centering
    \includegraphics[width=.85\columnwidth]{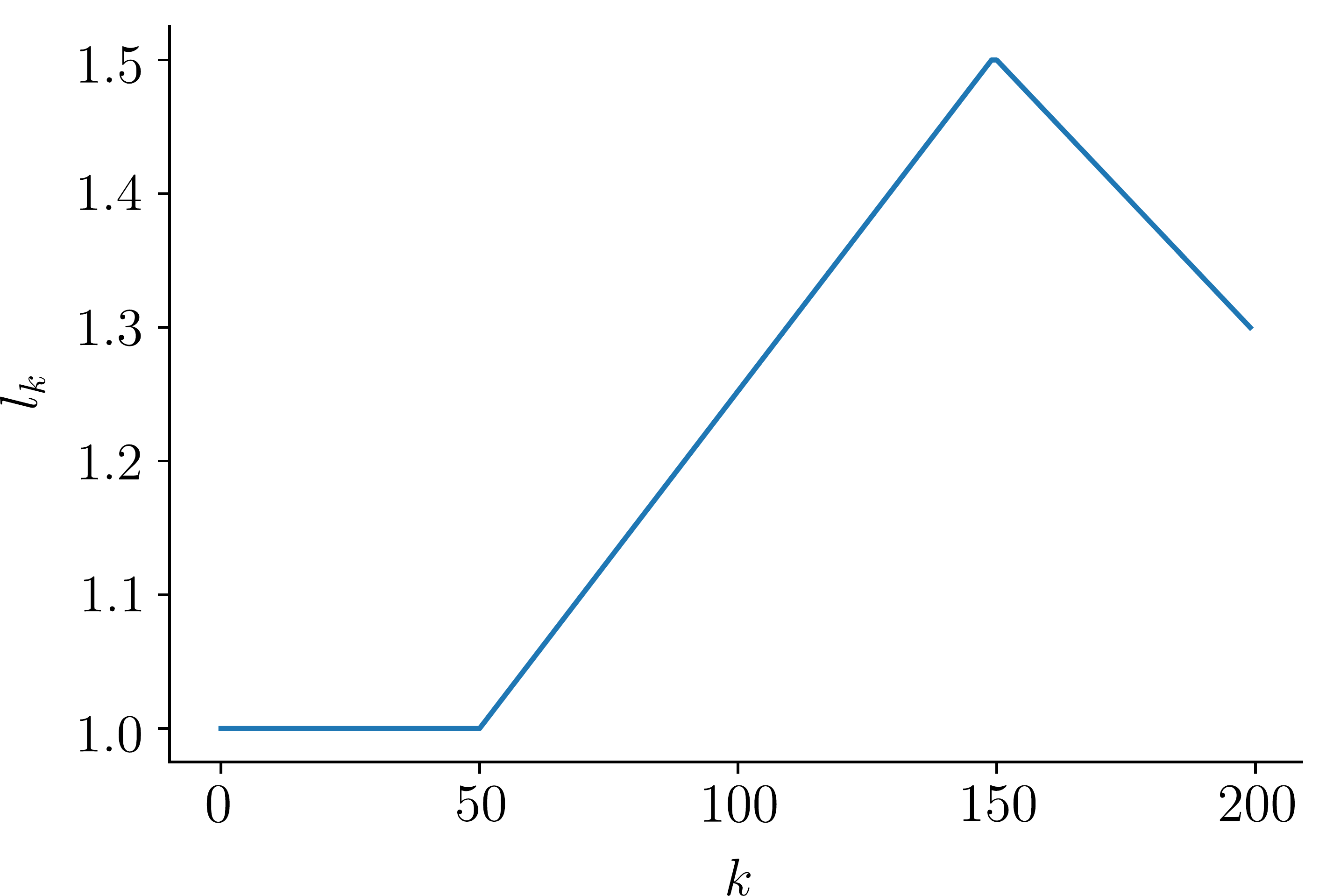}
    \caption{Evolution of the length of the crane cable $l_k$ as a function of the time index. }
    \label{fig:length_cable}
\end{figure}
The IMU and the three filters operate at the same frequency of $100$~Hz. The initial error covariance matrix is set to
\begin{equation}
    P_{0 \given 0} 
    = 
    \begin{bmatrix}
        0.05^2 & \zeros{1}{4}\\
        \zeros{4}{1} & 0.5^2 \identity{4}
    \end{bmatrix}.
\end{equation}
The gyroscope and accelerometer of the IMU are affected by Normal noise of zero mean and covariance
\begin{align}
    E\left( (w_k^{\omega})^2 \right) &= (0.005)^2,\\ 
    E\left( w_k^{a}(w_k^{a})^T\right) &= (0.005)^2\identity{2},
\end{align}
where $E(\cdot)$ denotes the expectation operator. The error  
\begin{equation}
    \xi_k = \log\left(\hat{\chi}_{k \given k}^{-1}\chi_k\right)
    \label{eq:nonlinear_error}
\end{equation}
is used to compare the performances of the three filters, where $\log(\cdot):\SE{2}{2} \rightarrow \mathbb{R}^5$ is the logarithmic map of $\SE{2}{2}$. The comparison is carried over $30$ simulations in which the initial error $\xi_0$ is drawn randomly from the distribution $\mathcal{N}(0_{5\times 1}, P_{0 \given 0})$. The average and standard deviations of the norm of the error $\xi_k$ are computed for the three filters. Results are plotted in \autoref{fig:results}.
\begin{figure}[thpb]
    \centering
    \includegraphics[width=\columnwidth]{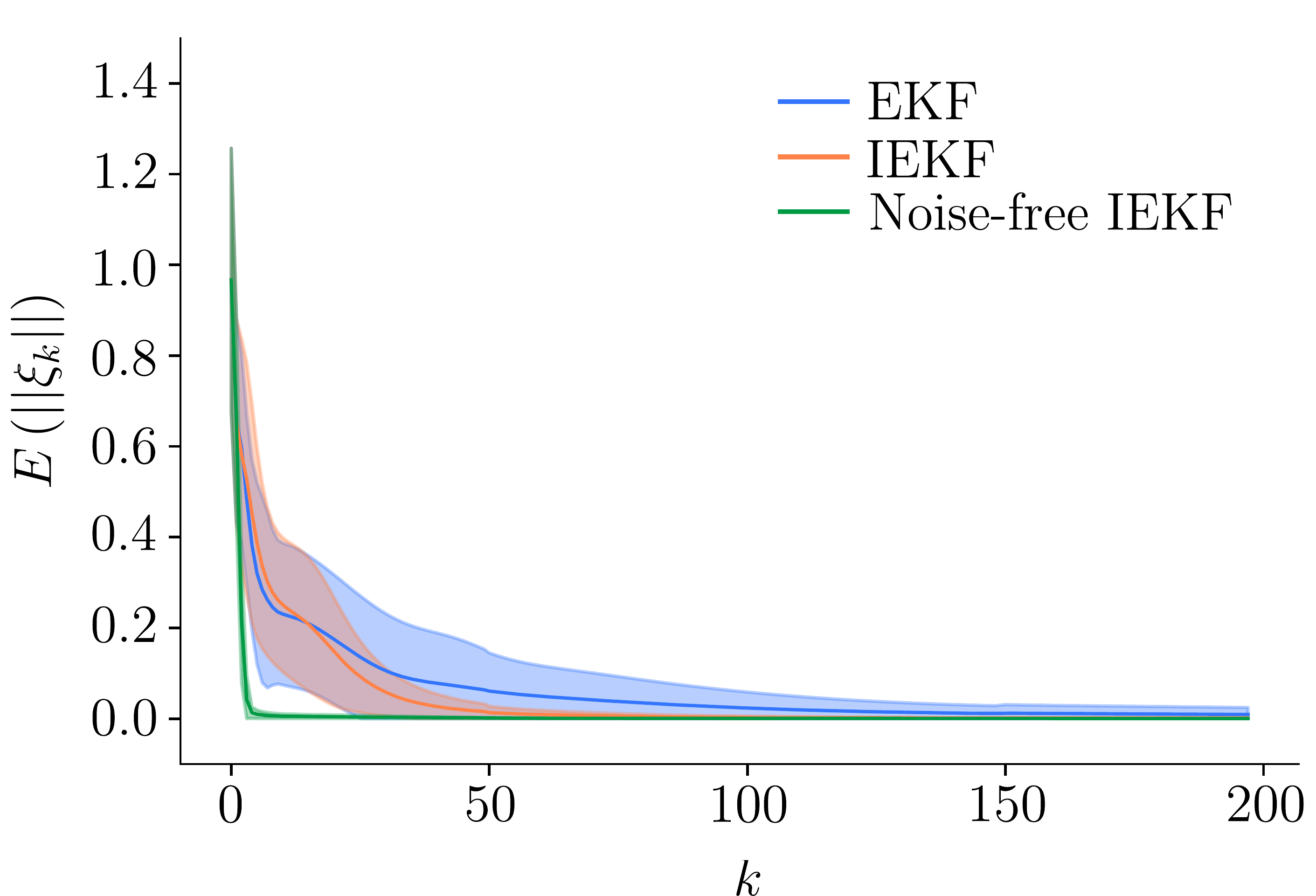}
    \caption{Mean norm of the error function $\xi_k$ as a function of the time index, computed over $30$ simulations. The standard deviation of the norm is displayed in light colors.}
    \label{fig:results}
\end{figure}
Two general comments can be made:
\begin{enumerate}
    \item The conventional extended Kalman filter is outperformed by its invariant filtering counterparts: its convergence is way slower, with a large variability across simulations.
    \item The noise-free IEKF exhibits the best convergence rate, with only $6$ time steps needed on average for $ \| \xi_k \| $ to go bellow $1\%$ of the initial error $ \| \xi_0 \| $, against $59$ and $191$ for the classical IEKF and EKF. It is also the method with the lowest variability from one simulation to another. Note that with a tolerance $\mathrm{tol} = 10^{-7}$, Algorithm~\ref{algo:iter_upd} performs $2.03$ cycles on average before $ \| z_{k \given k} \| $ stabilizes.
\end{enumerate}

This simulation shows that treating equality constraints like \eqref{eq:gen_cst} as proper noise-free pseudo-measurements clearly benefits the estimation and results in extremely fast convergence. 
\begin{remark}
    We indicate the failure to ensure Property~1 can sometimes make the noise-free IEKF diverge, when the residual innovation is too large. This leaves room for improvement, which is the subject of ongoing work.
\end{remark}


\section{Conclusion} \label{sec:conclu}
This work presented a practical and efficient solution for incorporating equality constraints of the form \eqref{eq:gen_cst} in the (extended) Kalman filtering framework. 
We advocated such constraints should be expressed as noise-free pseudo-measurements. 
We also derived a formulation of the Kalman gain that solves singular matrix inversion issues that may be encountered when dealing with noise-free pseudo-measurements. Using a proper matrix Lie group embedding and the theory of invariant filtering  \cite{barrau2018invariant} leads to Normal filters that correctly encode the actual physical uncertainty in the problem: their belief is consistent with the constraint in that the covariance matrix after update corresponds to an uncertainty that is wholly contained within the constrained subset. 
Finally, an alternative update for the IEKF was proposed in order to mitigate the impact of linearization errors in the noise-free update process. 
The performance of this method was evaluated on the task of estimating the extended pose of the hook of a crane, and was proved to outperform both the EKF and IEKF in this context.

As a perspective, we would like to thoroughly treat the case of the 3D crane, which poses some potential observability issues. We also intend to derive a general theory for noise-free measurements in the context of invariant filtering, namely for two-frame systems \cite{barrau2022geometry}. In this context, we would like to explore how to fully enforce the constraint $y_k=h(\hat{\chi}_{k \given k})$ at update in a natural way, that is, while retaining all the information that has been acquired before update. This might be achieved via the iterated Kalman filter \cite{bell1993iterated,bourmaud2016intrinsic}. To do iterated EKF in the present context, we anticipate the results developed at Theorem \ref{thm22} may prove useful. 


\appendices
\section{Kalman filter}\label{app:KF}
Let $x_k\in \mathbb{R}^n$ be the state to estimate. The discrete-time Kalman Filter (KF) is an optimal probabilistic filter that fuses information coming from two noisy sources: the state dynamics model
\begin{equation}
    x_{k+1} = F_k x_k + B_k u_k + w_k,
\end{equation}
with $F_k \in \mathbb{R}^{n\times n}$, $B_k\in \mathbb{R}^{n\times b}$, input $u_k \in \mathbb{R}^{b}$, and process noise $w_k \in \mathbb{R}^n$ with covariance $Q_k$, and measurements 
\begin{equation}
    y_k = H_k x_k + n_k
\end{equation}
with $H_k \in \mathbb{R}^{m\times n}$ and measurement noise $n_k\in \mathbb{R}^m$ with covariance $N_k$. The KF makes the assumption that $\xi_{k \given j} = x_k - \hat{x}_{k \given j} \sim \mathcal{N}(\zeros{n}{1}, P_{k \given j})$, where $\hat{x}$, $\xi$ and $P$ respectively denote the filter estimate, the estimation error and the associated covariance matrix. 
One iteration of the KF consists of two stages: the propagation and the update.
\paragraph{Propagation}
The KF propagates its belief through the state dynamics model as follows:
\begin{subequations}
\begin{align}
    \hat{x}_{k+1 \given k} &= F_k \hat{x}_{k \given k} + B_k u_k,\\
    P_{k+1 \given k} &= F_k P_{k \given k}F_k^T + Q_k.
\end{align}
\end{subequations}

\paragraph{Update}
Receiving measurement $y_k$, the KF updates its belief as follows:
\begin{subequations}
\begin{align}
    z_k &= y_k - H_k \hat{x}_{k \given k-1},\\
    S_k &= H_k P_{k \given k-1} H_k^T + N_k, \\
    K_k &= P_{k \given k-1}H_k^T S_k^{-1},\label{eq:Kalman:gain}\\
    \hat{x}_{k \given k} &= \hat{x}_{k \given k-1} + K_k z_k,\label{eq:stat:up}\\
    P_{k \given k} &= (\identity{n} - K_kH_k)P_{k \given k-1}\label{eq:ric:up} 
\end{align}
\end{subequations}
where $z_k$ is called the innovation, $S_k$ is the innovation covariance, and $K_k$ is the Kalman gain.


\bibliographystyle{IEEEtran}
\bibliography{root.bib}

\end{document}